\documentclass[11pt]{article}
\usepackage{latexsym, amssymb}

\newtheorem{Theorem}{Theorem}
\newtheorem{Lemma}{Lemma}

\newtheorem{Proposition}{Proposition}

\newtheorem{Example}{Example}
\newtheorem{Remark}{Remark}

\begin{document}

\title{Stopping Set Distributions of Some Linear Codes
\thanks{This research is supported in part by the National
Natural Science Foundation of China under the Grants 60972011, 60872025 and 10990011.
The material in this work was presented in part at the IEEE Information Theory Workshop, Chengdu, China, Oct. 2006.
}}

\author{
Yong Jiang, Shu-Tao Xia\thanks{Yong Jiang and Shu-Tao Xia are with
the Graduate School at Shenzhen of Tsinghua University, Shenzhen,
Guangdong 518055, P. R. China. E-mail: jiangy@sz.tsinghua.edu.cn,
xiast@sz.tsinghua.edu.cn} and Fang-Wei Fu\thanks{F.-W. Fu is with
the Chern Institute of Mathematics and the Key Laboratory of Pure
Mathematics and Combinatorics, Nankai University, Tianjin 300071,
P. R. China. Email: fwfu@nankai.edu.cn}}

\date{}
\maketitle

\begin{abstract}

Stopping sets and stopping set distribution of an low-density
parity-check code are used to determine the performance of this code
under iterative decoding over a binary erasure channel (BEC).
Let $C$ be a binary $[n,k]$ linear code with parity-check matrix $H$, where the rows of $H$ may be
dependent. A stopping set $S$ of $C$ with parity-check matrix $H$ is a subset of column indices of $H$ such
that the restriction of $H$ to $S$ does not contain a row of weight one.
The stopping set distribution $\{T_i(H)\}_{i=0}^n$ enumerates the number of
stopping sets with size $i$ of $C$ with parity-check matrix $H$. Note that stopping sets and stopping set distribution
are related to the parity-check matrix $H$ of $C$. Let $H^{*}$ be the parity-check matrix of $C$ which is
formed by all the non-zero codewords of its dual code $C^{\perp}$.
A parity-check matrix $H$ is called BEC-optimal if
$T_i(H)=T_i(H^*),\; i=0,1,\ldots, n$ and $H$ has the smallest
number of rows. On the BEC, iterative
decoder of $C$ with BEC-optimal parity-check matrix is an optimal decoder
with much lower decoding complexity than the exhaustive decoder. In this paper,
we study stopping sets, stopping set distributions and BEC-optimal parity-check matrices of
binary linear codes. Using finite geometry in combinatorics, we obtain BEC-optimal parity-check matrices and then determine the
stopping set distributions for the Simplex codes, the Hamming codes, the first order
Reed-Muller codes and the extended Hamming codes.
\end{abstract}


\baselineskip=20pt

\section{Introduction}

It is well known that the performance of an low-density
parity-check (LDPC) code under iterative decoding over a binary
erasure channel (BEC) is completely determined by certain
combinatorial structures, called {\it stopping sets}, of the
parity-check matrix of the LDPC code \cite{dptru}\cite{sv}.
The weight distribution of a linear code plays an important role in
determining the performance of this linear code under maximum likelihood decoding
over a binary symmetric channel.
The so-called \emph{stopping set distribution} characterizes the performance of
an LDPC code under iterative decoding over BEC.
Stopping sets and stopping set distributions of linear codes have been studied recently
by a number of researchers, for examples, see \cite{aw}-\cite{koetter}, \cite{ks}-\cite{lm},
\cite{m}-\cite{sv} and \cite{w}-\cite{xf1}.

Let $C$ be a binary $[n,k,d]$ linear code with length $n$,
dimension $k$ and minimum distance $d$. Let $H$ be an $m\times n$
parity-check matrix of $C$, where the rows of $H$ may be
dependent. Let $I=\{1,2,\ldots, n\}$ and $J=\{1,2,\ldots, m\}$
denote the sets of column indices and row indices of $H$,
respectively. The \emph{Tanner graph} $G_H$ \cite{tanner} corresponding to $H$
is a bipartite graph comprising of $n$ variable nodes labelled by
the elements of $I$, $m$ check nodes labelled by the elements of $J$, and
the edge set $E\subseteq \{(i,j) : i\in I, j\in J\}$, where there
is an edge $(i,j)\in E$ if and only if $h_{ji}=1$. The
\emph{girth} $g$ of $G_H$, or briefly the girth of $H$, is defined
as the minimum length of circles in $G_H$. A {\it stopping set}
$S$ of $H$ is a subset of column indices $\{1,2,\ldots,n\}$ such
that the restriction of $H$ to $S$, say $H(S)$, does not contain a
row of weight one. The smallest size of a nonempty stopping set,
denoted by $s(H)$, is called the {\it stopping distance} of $C$.
The codewords with minimum weight $d$ are called the \emph{minimum
codewords} of $C$. Let $W(x)=\sum_{i=0}^n A_i x^i$ denote the
\emph{weight enumerator} of $C$, where $A_i$ is the number of
codewords with weight $i$. $\{A_i\}_{i=0}^n$ is called the
\emph{weight distribution} of $C$. The stopping sets with size
$s(H)$ are called the \emph{smallest stopping sets} of $H$. Let
$T^{(H)}(x)=\sum_{i=0}^n T_i(H) x^i$ denote the \emph{stopping set
enumerator} of $C$ with parity-check matrix $H$, where $T_i(H)$ is
the number of stopping sets of $H$ with size $i$. Note that
$\emptyset$ is defined as a stopping set and $T_0(H)=1$.
$\{T_i(H)\}_{i=0}^n$ is called the \emph{stopping set
distribution} (SSD) of $C$ with parity-check matrix $H$.
Note that the stopping sets and stopping set distribution
dependent on the choice of the parity-check matrix $H$ of $C$.

Schwartz and Vardy \cite{sv} defined the \emph{stopping
redundancy} of the binary linear code $C$ as the minimum number of
rows of $H$ such that $s(H)=d$. Etzion \cite{e} studied the
stopping redundancy of Reed-Muller codes. In particular, the
stopping redundancies are determined respectively for the Hamming
codes \cite{sv}, the Simplex codes and the extended Hamming codes
\cite{e}, and an upper bound on the stopping redundancy of the
first order Reed-Muller codes was obtained in \cite{e}. In this
paper, we study a similar concept of the binary linear code $C$,
BEC-optimal parity-check matrix, in which both the number of stopping sets and the number of rows
are minimal among all parity-check matrices of $C$.

Suppose a codeword $\mathbf{c}=(c_1,c_2,\ldots,c_n)\in C$ is
transmitted over the BEC. Let
$\mathbf{r}=(r_1,r_2,\ldots,r_n)$ be the received word. The
erasure set is defined by $E_{\mathbf{r}}=\{ j: r_j\ne 0,1\}$. An
\emph{incorrigible set} of $C$ is an erasure set which contains
the support of a non-zero codeword of $C$. As noted by Weber and Abdel-Ghaffar
in \cite{wa2}, the received word $\mathbf{r}$ can be decoded unambiguously if and only if it
matches exactly one codeword of $C$ on all its nonerased positions.
This is equivalent to the condition that the
erasure set $E_{\mathbf{r}}$ is not an incorrigible set since $C$ is a linear code.
A decoder is said to be \emph{optimal} for the
BEC if it can achieve unambiguous decoding whenever the
erasure set is not incorrigible.
Note that an exhaustive decoder searching the complete set of codewords
is optimal. Let $H^{*}$ be formed by rows
which are all the non-zero codewords of the dual code $C^{\perp}$,
and denote its stopping set enumerator by $T^*(x)=\sum_{i=0}^n
T_i^* x^i$. The iterative decoder with parity-check matrix $H^{*}$
achieves the best possible performance, but has the highest decoding complexity.
It is also known from \cite{wa2} and \cite{ht1} that the
iterative decoder with parity-check matrix $H^{*}$ is an optimal
decoder for the BEC. For fixed parity-check matrix $H$, since $H$ is a
sub-matrix formed by some rows of $H^*$, any stopping set of $H^*$
is a stopping set of $H$, but the converse proposition may not be
true in general. Hence, we have $T_i(H)\geq T_i^*$ for every $0\leq i\leq n$.
A parity-check matrix $H$ is called
\emph{BEC-optimal} if $T^{(H)}(x)=T^*(x)$ and $H$ has the smallest
number of rows. Since a BEC-optimal parity-check matrix has the
same SSD with $H^*$, the iterative decoder with BEC-optimal
parity-check matrix must be an optimal decoder and it has lower
decoding complexity than $H^*$. Moreover, it achieves the best possible performance
as the iterative decoder with parity-check matrix $H^{*}$.

For the binary $[2^m-1, 2^m-m-1,3]$ Hamming code, say
$\mathcal{H}(m)$, it is known from \cite{sv} that for any
parity-check matrix, the stopping distance is equal to the minimum
distance. In the 2004 Shannon lecture, McEliece \cite{m} gave an
exact expression for the number of smallest stopping sets of
$\mathcal{H}(m)$ with the full rank parity-check matrix $F$, i.e.,
\begin{eqnarray}
\label{i1} T_3(F)=\frac{1}{6} (5^m-3^{m+1}+2^{m+1}).
\end{eqnarray}
Recently, Abdel-Ghaffar and Weber \cite{aw} further determined the
whole SSD of $\mathcal{H}(m)$ with the parity-check matrix $F$.
From \cite{ms} we know that
\begin{eqnarray}
\label{i2} A_3=\frac{1}{3} (2^m-1)(2^{m-1}-1)
\end{eqnarray}
and $A_3< T_3(F)$, i.e., $F$ is not BEC-optimal. Weber and
Abdel-Ghaffar \cite{wa} showed that for the parity-check matrix
$H^{*}$, $T_3(H^{*})=A_3$ and $T_4(H^{*})=A_4$, but they did not
determine the whole SSD of $H^{*}$. In this paper, we obtain
BEC-optimal parity-check matrices and then determine their SSDs
for the Simplex codes, the Hamming codes, the first order
Reed-Muller codes and the extended Hamming codes by using finite
geometry theory. Moreover, the above BEC-optimal parity-check
matrices are unique up to the equivalence. The rest of this paper is
arranged as follows. In Section II, we give some notations and results
in combinatorics that are needed in this paper. In Section III,
we obtain the BEC-optimal matrices for the Simplex codes, the Hamming codes, the first order
Reed-Muller codes and the extended Hamming codes. In
Section IV, in order to determine the SSDs for these BEC-optimal parity-check matrices,
the stopping generators of finite geometries are introduced.
In Section V, we determine the SSDs for the corresponding BEC-optimal parity-check matrices of
these codes. Finally, some conclusions are given in Section VI.

\section{Preliminaries}

In this section, we introduce some notations and results of finite geometry and
Gaussian binomial coefficients that will be used in this paper.
\subsection{Finite Geometries}

Let $\mathbb{F}_q$ be a finite field of $q$ elements and
$\mathbb{F}_q^m$ be the $m$-dimensional vector space over
$\mathbb{F}_q$, where $m\ge 2$.

Let $EG(m,q)$ be the $m$-dimensional Euclidean geometry over
$\mathbb{F}_q$ \cite[pp. 692-702]{ms}. $EG(m,q)$ has $q^m$ points,
which are vectors of $\mathbb{F}_q$. The $\mu$-flat in $EG(m,q)$
is a $\mu$-dimensional subspace of $\mathbb{F}_q^m$ or its coset.
A point is a $0$-flat, a \emph{line} is a $1$-flat, a \emph{plane}
is a $2$-flat, and an $(m-1)$-flat is called a \emph{hyperplane}.

Let $PG(m,q)$ be the $m$-dimensional projective geometry over
$\mathbb{F}_q$ \cite[pp. 692-702]{ms}. $PG(m,q)$ is defined in
$\mathbb{F}_q^{m+1}\setminus\{\mathbf{0}\}$. Two nonzero vectors
$\mathbf{p,p'}\in \mathbb{F}_q^{m+1}$ are said to be equivalent if
there is $\lambda\in \mathbb{F}_q$ such that $\mathbf{p}=\lambda
\mathbf{p'}$. It is well known that all equivalence classes of
$\mathbb{F}_q^{m+1}\setminus\{\mathbf{0}\}$ form points of
$PG(m,q)$. $PG(m,q)$ has $(q^{m+1}-1)/(q-1)$ points. The
$\mu$-flat in $PG(m,q)$ is simply the set of equivalence classes contained
in a $({\mu}+1)$-dimensional subspace of $\mathbb{F}_q^{m+1}$.
$0$-flat, $1$-flat, and $(m-1)$-flat are also called point, line
and hyperplane respectively.

In this paper, in order to present a unified approach, we use
$FG(m,q)$ to denote either $EG(m,q)$ or $PG(m,q)$. Let $n$ denote
the number of points of $FG(m,q)$. All points of $FG(m,q)$ are
indexed from $1$ to $n$. We will use $i$ to denote the $i$-th
point of $FG(m,q)$ for convenience if there is no confusion. For
any two different points $i,i'\in FG(m,q)$, there is one and only
one line, say $L(i,i')$, passing through them; for any three
distinct points $i,i',i''\in FG(m,q)$ which are not collinear,
there is one and only one plane, say $M(i,i',i'')$, passing
through them. For a set of points $\Pi\subseteq FG(m,q)$, let
$\chi(\Pi)=(x_1,x_2,\ldots, x_n)$ denote the \emph{incidence
vector} of $\Pi$, i.e., $x_i$=1 if $i\in \Pi$ and $x_i=0$
otherwise. For $u>0$, a \emph{$u$-set} means a set of $u$ points
of $FG(m,q)$. For a non-empty subset $S$ of $FG(m,q)$, define
$\langle S\rangle$ as the flat generated by the points in $S$,
i.e., $\langle S\rangle$ is the flat containing $S$ with the
minimum dimension. Clearly, $\langle S\rangle$ solely exists and
for any flat $F\supseteq S$, $\langle S\rangle\subseteq F$. The
next lemma is obvious.
\begin{Lemma}
\label{lem1} Let $\Pi$ be a non-empty subset of $FG(m,q)$. Then
$\Pi$ is a flat if and only if $\langle S \rangle\subseteq \Pi$
for any non-empty $S\subseteq \Pi$. Moreover,\\
{\rm(i)}. Let $\Pi\subseteq PG(m-1,2)$ and $|\Pi| \ge 2$. Then
$\Pi$ is a flat if and only if $L(i,i')\subseteq \Pi$ for any two
different points $i,i'\in
\Pi$;\\
{\rm(ii)}. Let $\Pi\subseteq EG(m,2)$ and $|\Pi|\ge 3$. Then $\Pi$
is a flat if and only if $M(i,i',i'')\subseteq \Pi$ for any three
distinct points $i,i',i''\in \Pi$.
\end{Lemma}

For $0\le \mu_1<\mu_2\le m$, there are $N(\mu_2,\mu_1)$
$\;\mu_1$-flats contained in a given $\mu_2$-flat and
$A(\mu_2,\mu_1)$ $\;\mu_2$-flats containing a given $\mu_1$-flat,
where for $EG(m,q)$ and $PG(m,q)$ respectively (see \cite{txla05})
\begin{eqnarray}
\label{fg1}
N_{EG}(\mu_2,\mu_1)&=& q^{\mu_2-\mu_1} \prod_{i=1}^{\mu_1} \frac{q^{\mu_2-i+1}-1}{q^{\mu_1-i+1}-1},\\
\label{fg2}
N_{PG}(\mu_2,\mu_1)&=& \prod_{i=0}^{\mu_1} \frac{q^{\mu_2-i+1}-1}{q^{\mu_1-i+1}-1},
\end{eqnarray}
\vspace{-0.5cm}
\begin{eqnarray}
\label{fg3}
A_{EG}(\mu_2,\mu_1)=A_{PG}(\mu_2,\mu_1)= \prod_{i=\mu_1+1}^{\mu_2}
\frac{q^{m-i+1}-1}{q^{\mu_2-i+1}-1}.
\end{eqnarray}
For $1\le \mu\le m$, let $n=N(m,0)$ and $J=N(m,\mu)$ be the
numbers of points and $\mu$-flats in $FG(m,q)$ respectively. The
points and $\mu$-flats are indexed from $1$ to $n$ and $1$ to $J$
respectively. Let
\begin{eqnarray}
\label{h} H=H_{FG}(m,\mu)=(h_{ji})_{J\times n}
\end{eqnarray}
be the \emph{point-$\mu$-flat incidence matrix}, where $h_{ji}=1$
for $1\le j\le J$ and $1\le i \le n$ if and only if the $j$th
$\mu$-flat contains the $i$th point. The rows of $H$ correspond to
all the $\mu$-flats in $FG(m,q)$ and have the same weight
$N(\mu,0)$. The columns of $H$ correspond to all the points and
have the same weight $A(\mu,0)$. The binary linear code with the
parity-check matrix $H$ is a class of LDPC codes based on finite geometries
\cite{txla05}\cite{klf}\cite{xf2}, denoted by $C_{FG}(m,\mu)$. Clearly, the
girth of $H$ is 6 if $\mu=1$ and 4 otherwise \cite{txla05}. Xia
and Fu \cite{xf2} proved that
\begin{eqnarray}
\label{fg4} d&\ge& s(H)\ge
A(\mu,\mu-1)+1=\frac{q^{m-\mu+1}-1}{q-1}+1.
\end{eqnarray}
Clearly, for $q=2$ and $2\le \mu\le m$, $C_{EG}(m,\mu)$ is the
$(\mu-1)$-th order Reed-Muller code $RM(m,\mu-1)$
\cite{ms}\cite{txla05}. Since the minimum distance of
$RM(m,\mu-1)$ is $2^{m-\mu+1}$, by (\ref{fg4}), the stopping
distance is equal to the minimum distance.

\subsection{Gaussian binomial coefficients}

For non-negative integers $m\le n$, let
\begin{eqnarray}
\label{gauss} \left[n\atop m\right]_q &=& \prod_{i=0}^{m-1}
\frac{q^{n-i}-1}{q^{m-i}-1}
\end{eqnarray}
denote the $q$-binomial coefficient or Gaussian binomial
coefficient \cite[pp.443-444]{ms}. In this paper, we will omit the
subscript $q$ when $q=2$. It is easy to check that
\begin{eqnarray}
\label{gauss1} \left[n\atop 0\right]_q = \left[n\atop
n\right]_q=1,\quad \left[n\atop m\right]_q = \left[n\atop
n-m\right]_q,\\
\label{gauss2} \left[n\atop m\right]_q \left[m\atop r\right]_q=
\left[n\atop r\right]_q \left[n-r\atop m-r\right]_q.
\end{eqnarray}
The well-known \emph{Cauchy Binomial Theorem} states that
\begin{eqnarray}
\label{cauchy} \prod_{i=1}^m (1+q^i x) &=& \sum_{i=0}^m
\left[m\atop i\right]_q q^{i(i+1)/2} x^i.
\end{eqnarray}

From now on, we will always assume that $q=2$. As usual, we define
${0\choose 0}=1$, ${i_2\choose i_1}=0$, $\left[0\atop 0\right]
=1$, $\left[i_2\atop i_1\right] =0$, $\sum_{i=i_1}^{i_2} a_i =0$
and $\prod_{i=i_1}^{i_2} a_i =1$ if $i_1>i_2$. Letting $x=-1/2$ in
(\ref{cauchy}), we have that
\begin{eqnarray}
\label{cauchy2} \sum_{i=0}^m \left[m\atop i\right] 2^{i(i-1)/2}
(-1)^i  &=& \delta_{m,0},
\end{eqnarray}
where $\delta_{m,n}=1$ if $m=n$ and $\delta_{m,n}=0$ otherwise.
It is easy to check by (\ref{fg1})-(\ref{fg3}) and
(\ref{gauss})-(\ref{gauss2}) that
\begin{eqnarray}
\label{Ngauss}N_{PG}(\mu_2,\mu_1)&=&\left[\mu_2+1\atop
\mu_1+1\right], \\
\label{Ngauss1}
N_{EG}(\mu_2,\mu_1)&=&2^{\mu_2-\mu_1}\left[\mu_2\atop \mu_1\right],\\
\label{Ngauss2}
A(\mu_2,\mu_1)&=&\left[m-\mu_1\atop \mu_2-\mu_1\right],\\
\label{Ngauss3} N(l,l-j) N(l-j,k) &=&\left[l-k\atop
j\right]N(l,k).
\end{eqnarray}

\section{BEC Optimal Parity-Check Matrices}

In this section, using finite geometry theory,
we obtain the BEC-optimal matrices for the Simplex codes,
the Hamming codes, the first order Reed-Muller codes and the extended Hamming codes.

The points of $PG(m-1,2)$ are simply the nonzero vectors of
$\mathbb{F}_2^m$. A $\mu$-flat of $PG(m-1,2)$ is simply the
nonzero linear combination of $\mu+1$ linearly independent points.
By (\ref{fg2}) and (\ref{fg3}), $PG(m-1,2)$ has $2^m-1$ points,
${(2^m-1)(2^{m-1}-1)}/{3}$ lines and $2^m-1$ hyperplanes.
Moreover, every line contains three points.

The points of $EG(m,2)$ are simply the vectors of $\mathbb{F}_2^m$.
A $\mu$-flat of $EG(m,2)$ is simply a $\mu$-dimensional subspace or
its coset. By (\ref{fg1}) and (\ref{fg3}), $EG(m,2)$ has $2^m$
points, $2^{m-1}(2^m-1)$ lines, $2^{m-2}(2^m-1)(2^{m-1}-1)/3$
planes and $2^{m+1}-2$ hyperplanes. Moreover, every line contains
two points, every plane contains 4 points.

Let $RM(m,r)$ be the $r$-th order binary Reed-Muller code
\cite[Ch. 13]{ms}. By puncturing a fixed coordinate from all
codewords of $RM(m,r)$, we obtain the punctured Reed-Muller code
$RM(m,r)^*$.
\begin{Lemma}
\label{lem2}{\rm \cite[p. 381, Th. 10]{ms}} The incidence vectors
of all the $(m-r-1)$-flats of $PG(m-1,2)$ generate $RM(m, r)^*$.
\end{Lemma}

\begin{Lemma}
\label{lemeg3}{\rm \cite[p. 385, Th. 12]{ms}} The incidence
vectors of all the $(m-r)$-flats of $EG(m,2)$ generate $RM(m, r)$.
\end{Lemma}

It is well known that $RM(m,m-2)$ is the binary $[2^m,2^m-m-1,4]$ extended
Hamming code, which is also denoted by $\hat\mathcal{H}(m)$;
$RM(m,1)$ is the dual code of $\hat\mathcal{H}(m)$ and a binary $[2^m,m+1,2^{m-1}]$ linear code;
$RM(m,m-2)^*$ is the binary $[2^m-1,2^m-m-1,3]$ Hamming code, which is denoted by
$\mathcal{H}(m)$; the shortened $RM(m,1)$, or the Simplex code
$\mathcal{S}(m)$, is the dual code of $\mathcal{H}(m)$ and a binary $[2^m-1,m,2^{m-1}]$ linear code.

In $PG(m-1,2)$, by (\ref{h}), let
\begin{eqnarray}
\label{h1} H^{(1)}=H_{PG}(m-1,1)
\end{eqnarray}
be the ${(2^m-1)(2^{m-1}-1)}/{3}\times (2^m-1)$ point-line
incidence matrix. Clearly, $H^{(1)}$ has uniform row weight $3$
and uniform column weight $2^{m-1}-1$ and girth 6. By (\ref{h}),
let
\begin{eqnarray}
\label{h2} H^{(2)}=H_{PG}(m-1,m-2)+J,
\end{eqnarray}
where $H_{PG}(m-1,m-2)$ is the $(2^m-1)\times (2^m-1)$
point-hyperplane incidence matrix and $J$ is a $(2^m-1)\times
(2^m-1)$ all-1 matrix. It is obvious that for any hyperplane $P$,
the incidence vector of $\bar P=PG(m-1,2)\setminus P$ is a row of
$H^{(2)}$ and vice versa. Clearly, $H^{(2)}$ has uniform row
weight $2^{m-1}$, uniform column weight $2^{m-1}$ and girth 4.
\begin{Lemma}
\label{lem3} $H^{(1)}$ is a parity-check matrix of
$\mathcal{S}(m)$ and the rows of $H^{(1)}$ form all minimum
codewords of $\mathcal{H}(m)$. $H^{(2)}$ is a parity-check matrix
of $\mathcal{H}(m)$ and the rows of $H^{(2)}$ form all nonzero
codewords of $\mathcal{S}(m)$.
\end{Lemma}
\begin{proof}
By Lemma \ref{lem2}, the lines of $PG(m-1,2)$ generate
$RM(m,m-2)^*$ or $\mathcal{H}(m)$, which implies that $H^{(1)}$ is
a parity-check matrix of $\mathcal{S}(m)$. Since the number of
weight 3 codewords of $\mathcal{H}(m)$ is exactly
${(2^m-1)(2^{m-1}-1)}/{3}$ \cite[p. 64, Cor. 16]{ms}, the rows of
$H^{(1)}$ form all minimum codewords of $\mathcal{H}(m)$.

For the second part, since there are $2^m-1$ rows in $H^{(2)}$ and
$\mathcal{S}(m)$ has $2^m-1$ non-zero codewords, it is enough to
show that every row of $H^{(2)}$ is orthogonal to all rows of
$H^{(1)}$. Let $\chi(\bar P)$ be a row of $H^{(2)}$, where $P$ is
a hyperplane of $PG(m-1,2)$. By \cite[p. 697, problem (8)]{ms}, any
line $L$ either intersects $P$ on a unique point or lies in $P$.
Since $L$ has three points, $L$ can intersect $P$ on either one or
three points, i.e., $L$ can only intersect $\bar P$ on zero or two
points, which implies that $\chi(L)$ is orthogonal to $\chi(\bar
P)$. This finishes the proof.
\end{proof}

In $EG(m,2)$, by (\ref{h}), let
\begin{eqnarray}
\label{h3} H^{(3)}=H_{EG}(m,2)
\end{eqnarray}
be the $2^{m-2}(2^m-1)(2^{m-1}-1)/3\times 2^m$ point-plane
incidence matrix. By Lemma \ref{lemeg3}, $H^{(3)}$ generates
$\hat\mathcal{H}(m)$, which implies that $H^{(3)}$ is a
parity-check matrix of $RM(m,1)$. Clearly, $H^{(3)}$ has uniform
row weight $4$ and uniform column weight $(2^m-1)(2^{m-1}-1)/3$
and girth 4. By (\ref{h}), let
\begin{eqnarray}
\label{h4} H^{(4)}=H_{EG}(m,m-1),
\end{eqnarray}
be the $(2^{m+1}-2)\times 2^m$ point-hyperplane incidence matrix.
By Lemma \ref{lemeg3}, $H^{(4)}$ generates $RM(m,1)$, which implies
that $H^{(4)}$ is a parity-check matrix of $\hat\mathcal{H}(m)$.
Clearly, $H^{(4)}$ has uniform row weight $2^{m-1}$, uniform
column weight $2^m-1$ and girth 4.

Hence, $H^{(1)}, H^{(2)}, H^{(3)}, H^{(4)}$ are respectively the
parity-check matrices of $\mathcal{S}(m)$, $\mathcal{H}(m)$,
$RM(m,1)$, $\hat \mathcal{H}(m)$, and their rows are formed by all
minimum codewords of the dual codes. For convenience, we list
the results in the next table, where $\chi(\cdot)$ denotes an
incidence vector, $L$ a line, $M$ a plane, $P$ a hyperplane, and
$\bar P=PG(m-1,2)\setminus P$.
\begin{eqnarray*}
\begin{array}{cclll}
\mathcal{S}(m) & PG(m-1,2) & H^{(1)} \mbox{ has rows formed by all
 } \chi(L) &H^{(1)*}\\
\mathcal{H}(m) & PG(m-1,2) & H^{(2)}  \mbox{ has rows formed by
all } \chi(\bar P)&H^{(2)*}\\
RM(m,1) & EG(m,2) & H^{(3)}  \mbox{ has rows formed by all  }
\chi(M)&H^{(3)*}\\
\hat\mathcal{H}(m) & EG(m,2) & H^{(4)}  \mbox{ has  rows formed by
all  } \chi(P)&H^{(4)*}
\end{array}
\end{eqnarray*}
Moreover, $H^{(1)*}, H^{(2)*}, H^{(3)*}, H^{(4)*}$ have rows
formed by all non-zero codewords of $\mathcal{H}(m)$,
$\mathcal{S}(m)$, $\hat \mathcal{H}(m)$, $RM(m,1)$, respectively.
Clearly, $H^{(2)}=H^{(2)*}$ and $H^{(4)}$ is formed by all rows except the all-$1$ row
of $H^{(4)*}$.

\begin{Proposition}
\label{prop1} Let $C$ be a binary linear code with parity-check
matrix $H$. Let $C^{\perp}$ be the dual code of $C$. The minimum distance $d^{\perp}$ of $C^{\perp}$
is at least $3$. Then a necessary condition of
$T^{(H)}(x)=T^*(x)$ is that all minimum codewords of $C^{\perp}$
are contained in rows of $H$.
\end{Proposition}
\begin{proof}
Assuming the contrary that there is a minimum codeword of
$C^{\perp}$, say $\mathbf{y_0}$, is not in the rows of $H$, it is
enough to show that there is a stopping set $S$ of $H$ such that
$S$ is not a stopping set of $H^*$. Fixing a coordinate $i_0\in
{\rm supp}(\mathbf{y_0})$, let $S=\overline{{\rm supp}(\mathbf{y_0})} \cup
\{i_0\}$, where
$\overline{{\rm supp}(\mathbf{y_0})}=\{1,2,\ldots,n\}\setminus
{\rm supp}(\mathbf{y_0})$. Since $S\cap {\rm supp}(\mathbf{y_0})=\{i_0\}$, $S$
is not a stopping set of $H^*$. On the other hand, for any
non-zero row $\mathbf{y}$ of $H$, we will show that $|S\cap
{\rm supp}(\mathbf{y})|\ge 2$ which implies that $S$ is a stopping set
of $H$. Clearly, $\mathbf{y}$ is a non-zero codeword of
$C^{\perp}$ other than $\mathbf{y_0}$. We claim that
$|\overline{{\rm supp}(\mathbf{y_0})}\cap {\rm supp}(\mathbf{y})|\geq 2$.
Assume the contrary that
$|\overline{{\rm supp}(\mathbf{y_0})}\cap {\rm supp}(\mathbf{y})|\le 1$. Then
$$|{\rm supp}(\mathbf{y_0})\cap {\rm supp}(\mathbf{y})|=
|{\rm supp}(\mathbf{y})|-|\overline{{\rm supp}(\mathbf{y_0})}\cap
{\rm supp}(\mathbf{y})|\ge |{\rm supp}(\mathbf{y})|-1.$$ Clearly,
$d_H(\mathbf{y},\mathbf{y_0})\ge d^{\perp}$ and
$w_H(\mathbf{y})=|{\rm supp}(\mathbf{y})|\ge d^{\perp}=w_H(\mathbf{y_0})$.
Hence,
\begin{eqnarray*}
d_H(\mathbf{y},\mathbf{y_0})&=&w_H(\mathbf{y})+w_H(\mathbf{y_0})-2|{\rm supp}(\mathbf{y})\cap {\rm supp}(\mathbf{y_0})|\\
&\le & w_H(\mathbf{y})+w_H(\mathbf{y_0})-2(w_H(\mathbf{y})-1)\\
&=& w_H(\mathbf{y_0})-w_H(\mathbf{y})+2\le 2,
\end{eqnarray*}
which leads a contradiction. Hence,
$|\overline{{\rm supp}(\mathbf{y_0})}\cap {\rm supp}(\mathbf{y})|\ge 2$, which
implies that $|S\cap
{\rm supp}(\mathbf{y})|\ge|\overline{{\rm supp}(\mathbf{y_0})}\cap
{\rm supp}(\mathbf{y})|\ge 2$ and thus $S$ is a stopping set of $H$. This
completes the proof.
\end{proof}

\begin{Proposition}
\label{prop2} Let $C$ be a binary linear code with parity-check
matrix $H$. Then a sufficient condition of $T^{(H)}(x)=T^*(x)$ is
that for any non-zero stopping set $S$ of $H$,
\begin{eqnarray}
\label{sc}
S=\bigcup_{\mathbf{x}\in C, {\rm supp}(\mathbf{x})\subseteq S}
{\rm supp}(\mathbf{x}).
\end{eqnarray}
\end{Proposition}
\begin{proof}
Let $S$ be a stopping set of $H$ and $S=\bigcup_{\mathbf{x}\in C,
{\rm supp}(\mathbf{x})\subseteq S} {\rm supp}(\mathbf{x})$. We only need to
show that $S$ is also a stopping set of $H^*$, i.e., for any fixed
row of $H^*$, say $\mathbf{y}$,  $|S\cap {\rm supp}(\mathbf{y})|\ne 1$,
or
\begin{eqnarray}
\label{eq10} \left |\bigcup_{\mathbf{x}\in C,
{\rm supp}(\mathbf{x})\subseteq S}\Big[ {\rm supp}(\mathbf{x})\cap
{\rm supp}(\mathbf{y})\Big]\right|\ne 1.
\end{eqnarray}
Since $\mathbf{y}$ represents a parity-check equation of $C$,
${\rm supp}(\mathbf{x})\cap {\rm supp}(\mathbf{y})$ must have  even number
elements for any $\mathbf{x}\in C$. Thus (\ref{eq10}) holds, which
finishes the proof.
\end{proof}
\begin{Remark}
Suppose a parity-check matrix $H$ of $C$ is formed by all minimum
codewords of $C^{\perp}$ with $d^{\perp}\ge 3$. It is easy to see by Propositions
$\ref{prop1}$ and $\ref{prop2}$ that $H$ is BEC-optimal provided that
$H$ satisfies the condition of Proposition $\ref{prop2}$.
\end{Remark}

\begin{Lemma}
\label{lems1} Let $\mathcal{S}(m)$ be the $[2^m-1,m,2^{m-1}]$
Simplex code with parity-check matrix $H^{(1)}$. Then $S\subseteq
PG(m-1,2)$ is a stopping set if and only if $S=PG(m-1,2)$ or $\bar
S=PG(m-1,2)\setminus S$ is a flat of $PG(m-1,2)$.
\end{Lemma}
\begin{proof}
By the definition of stopping set, $S\subseteq PG(m-1,2)$ is a
stopping set if and only if $H^{(1)}(S)$ has no rows with weight
one, i.e., $|L\cap S|\ne 1$ for any line $L$. Since $L$ has only
three points, $|L\cap S|\ne 1$ is equivalent to $|L\cap \bar S|\ne
2$. Hence, $S$ is a stopping set if and only if any line $L$
intersects $\bar S$ on 0, or 1, or 3 points. Clearly, if $|\bar
S|\le 1$, this is equivalent to $S=PG(m-1,2)$ or $\bar S$ is a
0-flat. Otherwise, if $|\bar S|\ge 2$, this is equivalent to
$L(i,j) \in \bar S$ for any different $i,j \in \bar S$. Hence, the
lemma follows by (i) of Lemma \ref{lem1}.
\end{proof}

By using (ii) of Lemma \ref{lem1} and the similar arguments used in the proof of Lemma
\ref{lems1}, it is easy to obtain the next lemma.

\begin{Lemma}
\label{lemrm1} Let $RM(m,1)$ be the first order Reed-Muller code
with parity-check matrix $H^{(3)}$. Then $S\subseteq EG(m,2)$ is a
stopping set if and only if $S=EG(m,2)$ or $\bar
S=EG(m,2)\setminus S$ is a flat of $EG(m,2)$.
\end{Lemma}

\begin{Theorem}
\label{th1} $H^{(1)}, H^{(2)}, H^{(3)}, H^{(4)}$ are the
BEC-optimal parity-check matrices for $\mathcal{S}(m)$,
$\mathcal{H}(m)$, $RM(m,1)$, $\hat \mathcal{H}(m)$, respectively.
Moreover, for each of the above four cases, there is no other
BEC-optimal parity-check matrix up to the permutation of rows.
\end{Theorem}
\begin{proof}
Note that the rows of $H^{(1)}, H^{(2)}, H^{(3)}, H^{(4)}$
are formed by all minimum codewords of the dual codes of
$\mathcal{S}(m)$, $\mathcal{H}(m)$, $RM(m,1)$, $\hat \mathcal{H}(m)$, respectively.
By Proposition \ref{prop1}, we only need to show that $T^{(H)}(x)=T^*(x)$
for $H=H^{(1)}, H^{(2)}, H^{(3)}, H^{(4)}$.

(i) $H=H^{(1)}$:\quad We show that $H^{(1)}$ satisfies the sufficient condition given in
Proposition \ref{prop2}. Let $S$ be a non-empty stopping set of
$H^{(1)}$. We need to show that $S$ satisfies (\ref{sc}).
If $|S|= n$, it is true since there is no codewords of weight $1$ in the dual code of $\mathcal{S}(m)$.
If $1\le |S|\le n-1$, by Lemma \ref{lems1}, $\bar S$ is a $\mu$-flat of
$PG(m-1,2)$, where $0\le \mu\le m-2$. Let
$P_1,P_2,\ldots,P_{A(m-2,\mu)}$ be all the hyperplanes which
contain $\bar S$, then $\bar S=\bigcap_{j=1}^{A(m-2,\mu)} P_j$, or
$ S=\bigcup_{j=1}^{A(m-2,\mu)} \bar P_j$. Since every $\bar P_j$
is the support of a codeword of $\mathcal{S}(m)$, (\ref{sc}) holds for $S$.

(ii) $H=H^{(2)}$:\quad It follows from the fact that $H^{(2)}=H^{(2)*}$.

(iii) $H=H^{(3)}$:\quad It is totally similar to the case (i).

(iv) $H=H^{(4)}$:\quad It follows from the fact that $H^{(4)}$ is formed by all rows except the all-$1$ row
of $H^{(4)*}$.
\end{proof}

\section{Generators in Finite Geometries}

In this section, we introduce the concept of stopping generators of finite geometries and give some enumeration results
that will be used to determine the SSDs for the BEC-optimal parity-check matrices $H^{(1)}, H^{(2)}, H^{(3)}, H^{(4)}$
of $\mathcal{S}(m)$, $\mathcal{H}(m)$, $RM(m,1)$, $\hat \mathcal{H}(m)$.

Let $S$ be a non-empty subset of $FG(m,2)$. For any $j\in S$,
denote $S_j=S\setminus\{j\}$. A point $i$ is said to be
\emph{independent} to $S$ if $i\not\in \langle S\rangle$.
$S$ is said to be \emph{independent} if for any
$j\in S$, $j$ is independent to $\langle S_j\rangle$. The empty set
$\emptyset$ is defined as an independent set. It is known from
\cite{ms} that the dimension of a flat $F$ of $FG(m,2)$ is equal
to $|J|-1$, where $J$ is an independent subset of $F$ with maximum
size. Clearly, for a non-empty set $S$, $S$ is independent if and
only if $\langle S\rangle$ is an $(|S|-1)$-flat.

For an integer $0\le l\le m$, let $F^{(l)}$ denote an $l$-flat of
$FG(m,2)$. $F^{(l)}$ has $N(l,0)$ points. Let $u\ge 1$, if a
$u$-set generates $F^{(l)}$, we call it a \emph{$u$-generator} of
$F^{(l)}$. If a $u$-generator $S$ of $F^{(l)}$ satisfies $\langle
S_j\rangle= \langle S \rangle =F^{(l)}$ for any $j\in S$, we call
$S$ a \emph{stopping $u$-generator} of $F^{(l)}$. Define $B(u,l)$
as the number of $u$-generators of $F^{(l)}$ and $G(u,l)$ as the
number of stopping $u$-generators of $F^{(l)}$, i.e, for $u\ge 1$
and $l\ge 0$,
\begin{eqnarray}
\label{b} \!\!\!\!B(u,l)\!\!\!\!&=&\!\!\!\!|\{S\subseteq F^{(l)}:
|S|=u,\;\langle
S\rangle=F^{(l)}\}|,\\
\!\!\!\!G(u,l)\!\!\!\!&=&\!\!\!\!|\{S\subseteq F^{(l)}:
|S|=u,\;\forall j\in S, \langle S_j\rangle= F^{(l)} \}|.\label{g}
\end{eqnarray}
Define $B(u,l)=0$ if $u\le 0$ or $l<0$. Clearly,
\begin{eqnarray}
\label{g3} G(u,l)&\le& B(u,l),\\
\label{g4} B(u,l)&=&0 \quad \mbox{if } u\le l.
\end{eqnarray}
For a $u$-set $S$, where $u\ge 1$, $S$ is a $u$-generator of a
$(u-1)$-flat if and only if $S$ is independent. A non-empty
independent set $S$ could not be a stopping generator, this is
because for any $j\in S$, $\langle S_j\rangle \subset \langle
S\rangle$. Hence, $G(u,u-1)=0$ for any $u\ge 1$. Combining this fact
with (\ref{g3})-(\ref{g4}), we have
\begin{eqnarray}
\label{g5}
G(u,l)=0 \quad\mbox{if }u\le l+1.
\end{eqnarray}

\begin{Lemma}
\label{lem7} For any $u\ge 1$ and $l\ge 0$, $B(u,l)$ satisfies the
following recursive equation
\begin{eqnarray}
\label{b1}\!\!\!\!B(1,0)\!\!\!\!&=&\!\!1, \quad
B(u,0) = 0 \quad \mbox{ if } u\ge 2, \\
\label{b2}\!\!\!\!{N(l,0)\choose u}&=&\sum_{i=0}^l
N(l,i)\;B(u,i),\; l\ge 0.
\end{eqnarray}
\end{Lemma}
\begin{proof}
(\ref{b1}) is obvious by (\ref{b}) and (\ref{g4}). In $F^{(l)}$,
there are ${N(l,0)\choose u}$ $u$-subsets, and each of which
generates an $i$-flat, where $0\le i \le l$. There are $N(l,i)$
$\;i$-flats in $F^{(l)}$, and each of which contains $B(u,i)$
$u$-generators of this $i$-flat $F^{(i)}$. Clearly, these $u$-sets are distinct,
which implies the lemma.
\end{proof}

\begin{Lemma}
\label{lemb}
\begin{eqnarray}
\label{bb} B(u,l)&=&\sum_{j=0}^l (-1)^j 2^{j(j-1)/2} N(l,l-j)
{N(l-j,0)\choose u},\\
\label{bpg} B_{PG}(u,l)&=&\sum_{j=0}^l (-1)^j 2^{j(j-1)/2}
\left[l+1\atop j\right] {2^{l-j+1}-1\choose u},\\
\label{beg} B_{EG}(u,l)&=&\sum_{j=0}^l (-1)^j 2^{j(j+1)/2}
\left[l\atop j\right] {2^{l-j}\choose u}.
\end{eqnarray}
\end{Lemma}
\begin{proof}
By Lemma \ref{lem7},
\begin{eqnarray*}
{N(l-j,0)\choose u}&=&\sum_{k=0}^{l-j}N(l-j,k)B(u,k).
\end{eqnarray*}
Hence, by (\ref{Ngauss3}) and (\ref{cauchy2}),
\begin{eqnarray*}
&&\sum_{j=0}^l (-1)^j 2^{j(j-1)/2} N(l,l-j)
{N(l-j,0)\choose u}\\
&=& \sum_{j=0}^l \sum_{k=0}^{l-j}(-1)^j 2^{j(j-1)/2}
N(l,l-j) N(l-j,k)B(u,k)\\
&=& \sum_{k=0}^l \sum_{j=0}^{l-k}(-1)^j 2^{j(j-1)/2}
\left[l-k\atop j\right]N(l,k) B(u,k) \\
&=& \sum_{k=0}^l N(l,k) B(u,k)\cdot\sum_{j=0}^{l-k}(-1)^j
2^{j(j-1)/2}\left[l-k\atop
j\right]\\
&=& \sum_{k=0}^l N(l,k) B(u,k)\cdot\delta_{l-k,0} \\
&=& N(l,l) B(u,l)=B(u,l).
\end{eqnarray*}
Moreover, (\ref{bpg}) and (\ref{beg}) follow from (\ref{bb}) and
(\ref{Ngauss})-(\ref{Ngauss1}).
\end{proof}

\begin{Lemma}
\label{lem8} Let $l\ge 0$, $u\ge l+1$, and $S$ be a $u$-generator
of $F^{(l)}$, where $F^{(l)}$ is an $l$-flat. Let $J=\{j\in S:
j\not\in\langle S_j \rangle \}$. Then\\
{\rm(i)}\quad $J$ is an independent set;\\
{\rm(ii)}\quad $J=\emptyset$ if and only if $S$ is a stopping
$u$-generator of $F^{(l)}$;\\
{\rm(iii)}\quad $J=S$ if and only if $S$ is an independent set;\\
{\rm(iv)}\quad otherwise, suppose $J$ is a non-empty proper subset of
$S$ and $|J|=k$, then $\langle S\setminus J\rangle$ is an
$(l-k)$-flat,  $1\le k\le l-1$, and $S\setminus J$ is a stopping
$(u-k)$-generator of $\langle S\setminus J\rangle$.
\end{Lemma}
\begin{proof}
Note that $\emptyset$ is an independent set according to the definition. If
$J\ne \emptyset$, for any $j\in J$, $J\subseteq S$ implies
$J_j\subseteq S_j$ and $\langle J_j\rangle\subseteq \langle
S_j\rangle$. Hence, $j\in J$ implies $j\not\in \langle S_{j}
\rangle$ and $j \not\in \langle J_j\rangle$. This completes the
proof of (i). By the definition of stopping generator, (ii) is
obvious. (iii) follows from (i) and the definition of independent
set. Next, we suppose $J=\{j_1,j_2,\ldots,j_k\}$ $(k\ge 1)$ is a
non-empty proper subset of $S$ and give the proof of (iv).

Note that $\langle S\rangle =F^{(l)}$ is an $l$-flat. Since
$j_1\in \langle S\rangle$ and $j_1\not\in \langle S\setminus
\{j_1\}\rangle$, $\langle S\setminus \{j_1\}\rangle$ is an
$(l-1)$-flat in $\langle S\rangle$. Since $j_2\in
S\setminus\{j_1\}$, $j_2\in \langle S\setminus\{j_1\}\rangle$.
Moreover, $j_2\not\in\langle S\setminus \{j_1, j_2\}\rangle$ since
$j_2\not\in \langle S\setminus \{j_2\}\rangle$. Hence, $\langle
S\setminus \{j_1,j_2\}\rangle$ is an $(l-2)$-flat in $\langle
S\setminus\{j_1\}\rangle$.  Repeating the above procedure, we have
that $\langle S\setminus \{j_1,j_2,j_3\}\rangle$ is an
$(l-3)$-flat in $\langle S\setminus \{j_1,j_2\}\rangle$, $\ldots$,
$\langle S\setminus J\rangle$ is an $(l-k)$-flat in $\langle
S\setminus \{j_1,\ldots,j_{k-1}\}\rangle$. Since $S\setminus J$ is
non-empty, $l-k\ge 0$. If $k=l$, $\langle S\setminus J\rangle$ is
a single point set, say $\{i\}$. Then $S\setminus J=\{i\}$ or
$S_i=J$. Hence, by (i), $\langle S_i\rangle$ is an $(l-1)$-flat,
which implies $\langle S_i\rangle \subset \langle S\rangle$ and
$i\not\in \langle S_i\rangle$. This means $i\in J$ and leads a
contradiction. Therefore $1\le k\le l-1$.

Now, we show that $S\setminus J$ is a stopping generator, i.e.,
for any $j\in S\setminus J$, $\langle S\setminus J\rangle =
\langle S_j\setminus J\rangle$. Assume by contrary that there
exists $j^*\in S\setminus J$ such that $S_{j^*}\setminus J$
generates an $(l-k-1)$-flat in $\langle S\setminus J\rangle$. By
using the inverse procedure given in the last paragraph, it is not
difficult to see that $\langle S_{j^*}\setminus
\{j_1,\ldots,j_{k-1}\}\rangle$ is an $(l-k)$-flat, $\langle
S_{j^*}\setminus \{j_1,\ldots,j_{k-2}\}\rangle$ is an
$(l-k+1)$-flat, $\ldots$, $\langle S_{j^*}\setminus
\{j_1\}\rangle$ is an $(l-2)$-flat, $\langle S_{j^*}\rangle$ is an
$(l-1)$-flat, which implies that $j^*\not\in \langle
S_{j^*}\rangle$ or $j^*\in J$. This gives a contradiction.

Combining these results, the lemma follows.
\end{proof}

\begin{Lemma}
\label{lem9} For any $l\ge 1$ and $0\le k\le l$, let $F^{(l)}$ be
an $l$-flat. Then there are exactly $\alpha(l,k)$ pairs
$(F^{(l-k)},J^{(k)})$ such that $F^{(l-k)}\subseteq F^{(l)}$ is an
$(l-k)$-flat, $J^{(k)}\subseteq F^{(l)}$ is an independent
$k$-set, and $\langle J^{(k)}\cup F^{(l-k)}\rangle=F^{(l)}$, where
\begin{eqnarray}
\label{alpha}\alpha(l,k)&=&\frac{N(l,l-k)}{k!} \prod_{i=1}^{k}
[N(l,0)-N(l-k+i,0)],\\
\label{alphapg} \alpha_{PG}(l,k)&=&\frac{1}{k!}\prod_{i=1}^k
2^{l-i+1}(2^{l-i+2}-1),\\
\label{alphaeg} \alpha_{EG}(l,k)&=&\frac{1}{k!}\prod_{i=1}^k
2^{l-i+1}(2^{l-i+1}-1).
\end{eqnarray}
\end{Lemma}
\begin{proof}
Clearly, $\alpha(l,0)=1$ which implies that (\ref{alpha}) holds
for $k=0$. It is easy to verify (\ref{alphapg}) and
(\ref{alphaeg}) from (\ref{alpha}) and (\ref{fg1})-(\ref{fg2}).
Hence, it is enough to show (\ref{alpha}) for $1\le k\le l$.
Suppose $F^{(l-k)}\subseteq F^{(l)}$ is a fixed $(l-k)$-flat. We
enumerate all suitable independent $k$-set $J^{(k)}$ as follows.
Choosing the first point from $F^{(l)}\setminus F^{(l-k)}$, there
are $N(l,0)-N(l-k,0)$ choices. $F^{(l-k)}$ and the first point
generate an $(l-k+1)$-flat, say $F^{(l-k+1)}$. Choosing the second
point from $F^{(l)}\setminus F^{(l-k+1)}$, there are
$N(l,0)-N(l-k+1,0)$ choices. $F^{(l-k+1)}$ and the second point
generate an $(l-k+2)$-flat, say $F^{(l-k+2)}$. Repeating the above
procedure, we have $N(l,0)-N(l-1,0)$ choices when choosing the $k$-th
point. It is easy to see that there are exactly $k!$ repetitions
for the above choosing procedure. Hence, there are totally
\begin{eqnarray*}
\frac{1}{k!} \prod_{i=0}^{k-1} [N(l,0)-N(l-k+i,0)]
\end{eqnarray*}
independent sets $J^{(k)}=\{j_1,j_2,\ldots, j_k\}$ to form a
suitable pair $(F^{(l-k)},J^{(k)})$ for fixed $(l-k)$-flat
$F^{(l-k)}$. Hence, (\ref{alpha}) follows from the fact that there
are $N(l,l-k)\;$ $(l-k)$-flats in $F^{(l)}$.
\end{proof}

\begin{Lemma}
\label{lem10} For any $u\ge 1$ and $l\ge 0$, $G(u,l)$ satisfies
the following recursive equation
\begin{eqnarray}
\label{g1} G(u,0)&=& 0 \;\mbox{ for any } u;\quad\quad
G(u,l)= 0 \;\mbox{ for any } u\le l+1;\\
B(u,l)&=&\sum_{k=0}^{l-1} \alpha(l,k)G(u-k,l-k),\quad u\ge l+2.
\label{g2}
\end{eqnarray}
\end{Lemma}
\begin{proof}
It is easy to check that (\ref{g1}) holds by the definition (\ref{g}) of
$G(u,l)$ and (\ref{g5}). Below we suppose $l\ge 1$ and $u\ge l+2$.
Since $u\ge l+2$, by Lemma \ref{lem8}, each $u$-generator of
$F^{(l)}$ is 1-1 corresponding to a $(u-k)$-stopping generator of
an $(l-k)$-flat of $F^{(l)}$, where $0\le k\le l-1$. For fixed
$0\le k\le l-1$, by Lemma \ref{lem9}, there are
$\alpha(l,k)G(u-k,l-k)$ such $u$-generators of $F^{(l)}$. Hence,
(\ref{g2}) follows by counting these $u$-generators where $k$ is from $0$
to $l-1$.
\end{proof}

\begin{Lemma}
\label{lemg} Let $u\ge l+2$. Then
\begin{eqnarray}
\label{gg} G(u,l)&=&\sum_{k=0}^{l-1} (-1)^k
\alpha(l,k)B(u-k,l-k),\\
\label{gpg} G_{PG}(u,l)&=& \sum_{k=0}^{l-1} B_{PG}(u-k,l-k)
\frac{(-1)^k}{k!}\prod_{i=1}^k 2^{l-i+1}(2^{l-i+2}-1),\\
\label{geg} G_{EG}(u,l)&=&\sum_{k=0}^{l-1} B_{EG}(u-k,l-k)
\frac{(-1)^k}{k!}\prod_{i=1}^k 2^{l-i+1}(2^{l-i+1}-1).
\end{eqnarray}
\end{Lemma}
\begin{proof}
It is easy to check by (\ref{alphapg})-(\ref{alphaeg}) that
\begin{eqnarray*}
\alpha(l,0)&=& 1,\\ \alpha(l,k)\alpha(l-k,j-k) &=& {j\choose
k}\alpha(l,j).
\end{eqnarray*}
Clearly, $\sum_{k=0}^{j} (-1)^k{j\choose k} = \delta_{j,0}$.
Moreover, by Lemma \ref{lem10},
\begin{eqnarray*}
B(u-k,l-k)&=&\sum_{j=0}^{l-k-1}
\alpha(l-k,j)G(u-k-j,l-k-j)\\
&=&\sum_{j=k}^{l-1} \alpha(l-k,j-k)G(u-j,l-j).
\end{eqnarray*}
Hence, using these equations, we have
\begin{eqnarray*}
&&\sum_{k=0}^{l-1} (-1)^k
\alpha(l,k)B(u-k,l-k)\\
&=& \sum_{k=0}^{l-1}\sum_{j=k}^{l-1} (-1)^k \alpha(l,k)
\alpha(l-k,j-k)G(u-j,l-j)\\
&=& \sum_{j=0}^{l-1} \sum_{k=0}^{j} (-1)^k
{j\choose k}\alpha(l,j)G(u-j,l-j)\\
&=& \sum_{j=0}^{l-1} \delta_{j,0}\alpha(l,j)G(u-j,l-j) \\
&=& \alpha(l,0)G(u,l)=G(u,l).
\end{eqnarray*}
Moreover, (\ref{gpg}) and (\ref{geg}) follow from (\ref{gg}) and
(\ref{alphapg})-(\ref{alphaeg}).
\end{proof}

\section{Stopping Set Distributions}

In this section, we determine the SSDs for the Simplex codes $\mathcal{S}(m)$, the Hamming codes $\mathcal{H}(m)$, the first order
Reed-Muller codes $RM(m,1)$ and the extended Hamming codes $\hat \mathcal{H}(m)$ with the BEC-optimal parity-check matrices $H^{(1)}, H^{(2)}, H^{(3)}, H^{(4)}$, respectively.

\subsection{Simplex Codes $\mathcal{S}(m)$}

Throughout this subsection, $n=2^m-1$ and
$PG(m-1,2)=\{1,2,\ldots,2^m-1\}$. By (\ref{fg2}),
there are $N_{PG}(m-1,\mu)$ $\;\mu$-flats in $PG(m-1,2)$ and a
$\mu$-flat has exactly $2^{\mu+1}-1$ points. The next theorem
follows from Lemma \ref{lems1} immediately.

\begin{Theorem}
\label{ths2} Let $\mathcal{S}(m)$ be the $[2^m-1,m,2^{m-1}]$
Simplex code with parity-check matrix $H^{(1)}$. Let
$\{T_i(H^{(1)})\}_{i=0}^n$ be the SSD of $\mathcal{S}(m)$. Then
\begin{eqnarray*}
T_i (H^{(1)}) = \left\{ \begin{array}{lll} 1 &
\mbox{if }\;i=0 \mbox{ or } 2^m-1,\\
N_{PG}(m-1,\mu), &\mbox{if }\;
i=2^m-2^{\mu+1}, \\
&\quad \mu=0,\ldots,m-2,\!\!\!\!\!\!\\
0, & \mbox{otherwise},
\end{array} \right.
\!\!\!\!\!\!
\end{eqnarray*}
where $$N_{PG}(m-1,\mu)= \prod_{i=0}^{\mu}
\frac{2^{m-i}-1}{2^{\mu-i+1}-1}.$$
\end{Theorem}
\begin{Remark}
Let $\mu=m-2$, by Theorem {\rm \ref{ths2}}, it is easy to check
that the number of smallest stopping sets
$T_{2^{m-1}}(H^{(1)})= 2^{m}-1,$ which coincides with the
number of minimum codewords of $\mathcal{S}(m)$.
\end{Remark}

\begin{Example}
\emph{By Theorem \ref{ths2}, we can easily calculate the SSDs
of $\mathcal{S}(m)$ with parity-check matrix $H^{(1)}$ by
\emph{Mathematica} software. Here are some examples for
$m=3,4,5$.\\
For $\mathcal{S}(3)$,
\begin{eqnarray*}
T(x)&=&1+7x^4+7x^6+x^7.
\end{eqnarray*}
For $\mathcal{S}(4)$,
\begin{eqnarray*}
T(x)&=&1+15x^8+35x^{12}+15x^{14}+x^{15}.
\end{eqnarray*}
For $\mathcal{S}(5)$,
\begin{eqnarray*}
T(x)&=&1+31x^{16}+155x^{24} +155x^{28}+31x^{30}+x^{31}.
\end{eqnarray*}
}
\end{Example}

It is worthy to note that all examples in this section besides the above one
are calculated through two ways, one of which uses the derived formula,
and the other of which uses the exhaust computer search for verification.

\subsection{Hamming Codes $\mathcal{H}(m)$}

Throughout this subsection, $n=2^m-1$ and
$PG(m-1,2)=\{1,2,\ldots,2^m-1\}$. Note that
$H^{(2)}=H^{(2)*}$ and $P$ is a hyperplane if and only if
$\chi(\bar P)$ is a row of $H^{(2)}$.

\begin{Lemma}
\label{lemh1} Let $\mathcal{H}(m)$ be the $[2^m-1,2^m-m-1,3]$
Hamming code with parity-check matrix $H^{(2)}$. Then $S\subseteq
PG(m-1,2)$ is a non-empty stopping set if and only if $\langle
S\rangle=\langle S_j\rangle$ for any $j\in S$, where
$S_j=S\setminus\{j\}$.
\end{Lemma}
\begin{proof}
By the definition of stopping sets, a non-empty subset $S\subseteq
PG(m-1,2)$ is a stopping set if and only if $H^{(2)}(S)$ has no
rows with weight one, i.e., $|\bar P\cap S|\ne 1$ for any
hyperplane $P$ of $PG(m-1,2)$. Clearly, $|\bar P\cap S|\ne 1$ is
equivalent to $|P\cap S|\ne |S|-1$. Hence, we only need to show
that $|P\cap S|\ne |S|-1$ for any hyperplane $P$ of $PG(m-1,2)$ if
and only if $\langle S\rangle=\langle S_j\rangle$ for any $j\in
S$.

Firstly, we will prove the necessary condition. Suppose that $S$
satisfies $|P\cap S|\ne |S|-1$ for any hyperplane $P$. Clearly,
$\langle S_j\rangle\subseteq\langle S\rangle$. Assume by contrary
that there exists $j\in S$ such that $\langle
S_j\rangle\subset\langle S\rangle$, i.e., $d_j=d-1$, where $d_j$
and $d$ are the dimensions of $\langle S_j\rangle$ and $\langle
S\rangle$ respectively. If $d=m-1$, then $\langle S_j\rangle$ is a
hyperplane not including $S$, i.e., $|\langle S_j\rangle\cap
S|=|S_j|= |S|-1$, which leads a contradiction. Otherwise, if
$d<m-1$, by (\ref{fg3}), there are $A(m-2, d)$ hyperplanes
containing $S$, and there are $A(m-2, d_j)$ hyperplanes containing
$S_j$. It is easy to check that for $PG(m-1,2)$
\begin{eqnarray*}
\frac{A(m-2, d_j)}{A(m-2, d)}=\prod_{i=d_j+1}^{d}
\frac{2^{m-i}-1}{2^{m-i-1}-1}=\frac{2^{m-d}-1}{2^{m-d-1}-1}>1,
\end{eqnarray*}
which implies that there exists a hyperplane, say $P^*$, such that
$S_j \subseteq P^*$ and $S \not \subseteq P^*$. Hence, $|P^*\cap
S|= |S|-1$, which leads a contradiction.

On the other hand, suppose that $S$ satisfies $\langle
S\rangle=\langle S_j\rangle$ for any $j\in  S$. Assume by contrary
that there exists a hyperplane $P^*$ such that $|P^*\cap S|=
|S|-1$, i.e., there exists a point $j^*\in S$ such that $P^*\cap
S=S_{j^*}$. Then $S_{j^*}\subseteq P^*$ and $S\not\subseteq P^*$,
i.e., $\langle S_{j^*}\rangle \subseteq P^*$ and $\langle S\rangle
\not\subseteq P^*$, which leads a contradiction.

Combining these claims, the lemma follows.
\end{proof}

\begin{Remark}
\label{rh} It is easy to see from Lemma $\ref{lemh1}$ that when
$u\ge 2^{m-1}+1$, any $u$-set is a stopping set since any set with
at least $2^{m-1}$ points generates $PG(m-1,2)$.
\end{Remark}

\begin{Theorem}
\label{thh2} Let $\mathcal{H}(m)$ be the $[2^m-1,2^m-m-1,3]$
Hamming code with parity-check matrix $H^{(2)}$. Let
$\{T_i(H^{(2)})\}_{i=0}^n$ be the SSD of $\mathcal{H}(m)$. Then
\begin{eqnarray}
\label{thh2-1}
T_u (H^{(2)}) = \left\{ \begin{array}{ll} 1, &u=0,\\
0, &u=1,2,\\
\sum_{l=\lfloor\log u\rfloor}^{\min\{u-2,m-1\}}&\!\!\!\!\!\! N_{PG}(m-1,l)\;G_{PG}(u,l), \\
&u=3,\ldots,2^{m-1},\\
{2^m-1\choose u}, & u=2^{m-1}+1,\ldots, 2^m-1,\!\!\!\!\!\!
\end{array} \right.
\end{eqnarray}
where $N_{PG}(m-1,l)$ and $G_{PG}(u,l)$ are defined in {\rm
(\ref{fg2})} and {\rm(\ref{gpg})} respectively.
\end{Theorem}
\begin{proof}
Clearly, $T_0=1$. By Lemma \ref{lemh1} and the definition
(\ref{g}) of $G_{PG}(u,l)$, it is easy to see that
\begin{eqnarray}
\label{thh2-2}T_u  =\sum_{l=0}^{m-1} N_{PG}(m-1,l)\;G_{PG}(u,l).
\end{eqnarray}
Since any $u$-set in $PG(m-1,2)$ generates a flat with dimension
at least $\lfloor \log u \rfloor$,
\begin{eqnarray}
\label{thh2-3} B_{PG}(u,l)=G_{PG}(u,l)=0 \quad \mbox{if }l<\lfloor
\log u \rfloor.
\end{eqnarray}
Combining (\ref{thh2-2}), (\ref{thh2-3}) and (\ref{g5}), we have
that
\begin{eqnarray}
\label{thh2-4} T_u  =\!\!\!\!\sum_{l=\lfloor\log
u\rfloor}^{\min\{u-2,m-1\}} N_{PG}(m-1,l)\;G_{PG}(u,l), \;1\le u
\le 2^m-1.
\end{eqnarray}
Let $u=1,2$, we have $T_1=T_2=0$. Combining these facts and Remark
\ref{rh}, (\ref{thh2-1}) follows.
\end{proof}

\begin{Remark}
By Theorem {\rm \ref{thh2}}, we have that
\begin{eqnarray*}
T_3(H^{(2)})&=&(2^m-1)(2^{m-1}-1)/3,\\
T_4(H^{(2)})&=&(2^m-1)(2^{m-1}-1)(2^{m-2}-1)/3.
\end{eqnarray*}
It is easy to see from {\rm \cite{ms}} that $A_3=T_3(H^{(2)})$ and
$A_4=T_4(H^{(2)})$ for $\mathcal{H}(m)$, which were also obtained
by Weber and Abdel-Ghaffar {\rm \cite{wa}}.
\end{Remark}

\begin{Example}
\emph{By Theorem \ref{thh2}, we can easily calculate the SSDs for
$\mathcal{H}(m)$ by \emph{Mathematica} software. Here are some
examples for $m=3,4,5$.}\\
\emph{For} $\mathcal{H}(3)$,
\begin{eqnarray*}
T(x)&=&1+7x^3+7x^4+21x^5+7x^6+x^7.
\end{eqnarray*}
\emph{For} $\mathcal{H}(4)$,
\begin{eqnarray*}
T(x)&=&1+35x^3+105x^4+483x^5+2485x^6+5595x^7+6315x^8\\
&&+5005x^9+3003x^{10}
+1365x^{11}+455x^{12}+105x^{13}+15x^{14}+x^{15}.
\end{eqnarray*}
\emph{For} $\mathcal{H}(5)$,
\begin{eqnarray*}
T(x)&=&1+155x^3+1085x^4+8463x^5+88573x^6+798095x^7\\
&&+4909005x^8+16998075x^9+41869685x^{10}+83182827x^{11}\\
&&+140443485x^{12}+206027395x^{13}+265130445x^{14}+300532755x^{15}\\
&&+300539699x^{16}+265182525x^{17}+206253075x^{18}
+141120525x^{19}\\
&&+84672315x^{20}+44352165x^{21}+20160075x^{22}+7888725x^{23}\\
&&+2629575x^{24}
+736281x^{25}+169911x^{26}+31465x^{27}+4495x^{28}\\
&&+465x^{29}+31x^{30}+x^{31}.
\end{eqnarray*}
\end{Example}

\subsection{The First Order Reed-Muller Codes $RM(m,1)$}

Throughout this subsection, $n=2^m$ and
$EG(m,2)=\{1,2,\ldots,2^m\}$. By (\ref{fg1}), there
are $N_{EG}(m,\mu)$ $\;\mu$-flats in $EG(m,2)$ and a $\mu$-flat
has exactly $2^\mu$ points. The next theorem follows from Lemma
\ref{lemrm1} immediately.

\begin{Theorem}
\label{thrm2} Let $RM(m,1)$ be the first order Reed-Muller code
with parity-check matrix $H^{(3)}$. Let
$\{T_i(H^{(3)})\}_{i=0}^n$ be the SSD of $RM(m,1)$. Then
\begin{eqnarray}
T_i (H^{(3)}) = \left\{ \begin{array}{lll} &1 &
\mbox{if }\;i=0 \mbox{ or } 2^m,\\
&N_{EG}(m,\mu), &\mbox{if }
\;i=2^m-2^\mu,\\
&&\quad\mu=0,1,\ldots,m-1,\!\!\!\!\\
&0, & \mbox{otherwise},
\end{array} \right.
\end{eqnarray}
where
\begin{eqnarray*}
N_{EG}(m,\mu)= 2^{m-\mu} \prod_{i=1}^{\mu}
\frac{2^{m-i+1}-1}{2^{\mu-i+1}-1}.
\end{eqnarray*}
\end{Theorem}

\begin{Remark}
Let $\mu=m-1$, by Theorem {\rm \ref{thrm2}}, it is easy to check
that the number of smallest stopping sets $T_{2^{m-1}}(H^{(3)})= 2^{m+1}-2,$
which coincides with the
number of minimum codewords of $RM(m,1)$.
\end{Remark}

\begin{Example}
\emph{By Theorem \ref{thrm2}, we can easily calculate the SSDs of $RM(m,1)$
with parity-check matrix $H^{(3)}$ by \emph{Mathematica} software.
Here are some examples for $m=3,4$.\\
For $RM(3,1)$,
\begin{eqnarray*}
T(x)&=&1+14x^4+28x^6+8x^7+x^8.
\end{eqnarray*}
For $RM(4,1)$,
\begin{eqnarray*}
T(x)&=&1+30x^8+140x^{12}+120x^{14}+16x^{15}+x^{16}.
\end{eqnarray*}
}
\end{Example}

\subsection{The Extended Hamming Codes $\hat \mathcal{H}(m)$}

Throughout this subsection, $n=2^m$ and
$EG(m,2)=\{1,2,\ldots,2^m\}$. Note that $P$ is a
hyperplane if and only if $\chi(P)$ is a row of $H^{(4)}$, and if
and only if $\bar P=EG(m,2)\setminus P$ is a hyperplane.

\begin{Lemma}
\label{lemeh1} Let $\hat\mathcal{H}(m)$ be the $[2^m,2^m-m-1,4]$
extended Hamming code with parity-check matrix $H^{(4)}$. Then
$S\subseteq EG(m,2)$ is a non-empty stopping set if and only if
$\langle S\rangle=\langle S_j\rangle$ for any $j\in S$.
\end{Lemma}
\begin{proof}
By the definition of stopping sets, a non-empty subset $S\subseteq EG(m,2)$ is a
stopping set if and only if $H^{(4)}(S)$ has no rows with weight
one, i.e., $|P\cap S|\ne 1$ or $|\bar P\cap S|\ne |S|-1$ for any
hyperplane $P$ of $EG(m,2)$. Since $P$ is a hyperplane in
$EG(m,2)$ if and only if $\bar P$ is also a hyperplane, we only
need to show that $|P\cap S|\ne |S|-1$ for any hyperplane $P$ of
$EG(m,2)$ if and only if $\langle S\rangle=\langle S_j\rangle$ for
any $j\in S$. With the same arguments used in the proof of Lemma \ref{lemh1}, the lemma
follows.
\end{proof}

\begin{Remark}
\label{reh} It is easy to see from Lemma $\ref{lemeh1}$ that when
$u\ge 2^{m-1}+2$, any $u$-set of $EG(m,2)$ is a stopping set since
any set with at least $2^{m-1}+1$ points generates $EG(m,2)$.
\end{Remark}

Since $H^{(4)}$ is formed by all rows except the all-$1$ row
of $H^{(4)*}$, they have the same SSDs.

\begin{Theorem}
\label{theh2} Let $\hat\mathcal{H}(m)$ be the $[2^m,2^m-m-1,4]$
extended Hamming code with parity-check matrix $H^{(4)}$. Let
$\{T_i(H^{(4)})\}_{i=0}^n$ be the SSD of $\hat\mathcal{H}(m)$.
Then
\begin{eqnarray}
\label{theh2-1}
T_u (H^{(4)}) = \left\{ \begin{array}{ll} 1, &u=0,\\
0, &u=1,2,3,\\
\sum_{l=\lceil\log u\rceil}^{\min\{u-2,m\}}&\!\!\!\!N_{EG}(m,l)\;G_{EG}(u,l), \\
&u=4,\ldots,2^{m-1}+1,\\
{2^m\choose u}, & u=2^{m-1}+2,\ldots, 2^m,
\end{array} \right.
\end{eqnarray}
where $N_{EG}(m,l)$ and $G_{EG}(u,l)$ are defined in $(\ref{fg1})$
and $(\ref{geg})$ respectively.
\end{Theorem}
\begin{proof}
Clearly, $T_0=1$. By Lemma \ref{lemeh1} and the definition
(\ref{g}) of $G_{EG}(u,l)$, it is easy to see that
\begin{eqnarray}
\label{theh2-2} T_u  =\sum_{l=0}^{m} N_{EG}(m,l)\;G_{EG}(u,l).
\end{eqnarray}
Since any $u$-set in $EG(m,2)$ generates a flat with dimension at
least $\lceil \log u \rceil$,
\begin{eqnarray}
\label{theh2-3} B_{EG}(u,l)=G_{EG}(u,l)=0 \quad \mbox{if }l<\lceil \log u
\rceil.
\end{eqnarray}
Combining (\ref{theh2-2})-(\ref{theh2-3}) and (\ref{g5}), we have
that
\begin{eqnarray}
\label{theh2-4} T_u  =\sum_{l=\lceil\log
u\rceil}^{\min\{u-2,\;m\}} N_{EG}(m,l)\;G_{EG}(u,l), \quad 1\le u \le
2^m.
\end{eqnarray}
Let $u=1,2,3$, we have $T_1=T_2=T_3=0$. Combining these results and Remark
\ref{reh}, (\ref{theh2-1}) follows.
\end{proof}

\begin{Remark}
By Theorem {\rm \ref{theh2}}, we have that $$T_4(H^{(4)})=
2^{m-2}(2^m-1)(2^{m-1}-1)/3, \quad T_5(H^{(4)})=0.$$ It is easy to see
from {\rm \cite{ms}} that $A_4=T_4(H^{(4)})$ and $A_5=0$ for $\hat
\mathcal{H}(m)$, which were also obtained by Weber and
Abdel-Ghaffar {\rm \cite{wa}}.
\end{Remark}

\section{Conclusions}

Let $C$ be a binary $[n,k]$ linear code. Let $H^{*}$ be the parity-check matrix of $C$ which is
formed by all the non-zero codewords of its dual code $C^{\perp}$.
On the BEC, the iterative decoder with parity-check matrix $H^{*}$
achieves the best possible performance, but has the highest decoding complexity.
The stopping set distribution of $C$ with the parity-check matrix $H^{*}$
is used to determine the performance of $C$ under iterative decoding with the parity-check matrix $H^{*}$ over a BEC.
In general, it is difficult to determine the stopping set distribution $\{T_i(H^*)\}_{i=0}^n$ of $C$ with the parity-check matrix $H^{*}$.
Let $H$ be a parity-check matrix of $C$. Let $\{T_i(H)\}_{i=0}^n$ be the stopping set distribution of $C$ with the parity-check matrix $H$.
Since $H$ is a sub-matrix formed by some rows of $H^*$, any stopping set of $H^*$
is a stopping set of $H$. This implies that $T_i(H)\geq T_i(H^*)$ for every $0\leq i\leq n$.
A parity-check matrix $H$ is called BEC-optimal if
$T_i(H)=T_i(H^*)$ for every $0\leq i\leq n$ and $H$ has the smallest number of rows.
On the BEC, the iterative decoder with BEC-optimal parity-check matrix $H$ achieves the best possible performance
as the iterative decoder with parity-check matrix $H^{*}$ and it has lower decoding complexity than $H^*$.
In general, it is difficult to obtain BEC-optimal parity-check matrix for a general linear code.
It is interesting to construct BEC-optimal parity-check matrices and then determine the corresponding
stopping set distributions for LDPC codes and well known linear codes.
In this paper, we obtain BEC-optimal parity-check matrices and then determine the corresponding
stopping set distributions for the Simplex codes, the Hamming codes, the first order
Reed-Muller codes and the extended Hamming codes.


\begin{thebibliography}{99}
\baselineskip=18pt

\bibitem{aw}
K. A. S. Abdel-Ghaffar and J. H. Weber, ``Complete enumeration of
stopping sets of full-rank parity-check matrices of Hamming codes,"
\emph{IEEE Trans. Inform. Theory}, vol. 53, no. 9, pp. 3196-3201, 2007.

\bibitem{dptru} C. Di, D. Proietti, I. E. Telatar, T. J. Richardson and R.L. Urbanke,
``Finite-length analysis of low-density parity-check codes on the
binary erasure channel," \emph{IEEE Trans. Inform. Theory}, vol.
48, no. 6, pp. 1570-1579, 2002.

\bibitem{ea}
M. Esmaeili and M. J. Amoshahy, ``On the stopping distance of array code parity-check matrices,"
\emph{IEEE Trans. Inform. Theory}, vol. 55, no. 8, pp. 3488-3493, Aug. 2009.

\bibitem{e}
T. Etzion, ``On the stopping redundancy of Reed-Muller codes,"
\emph{IEEE Trans. Inform. Theory}, vol. 52, no. 11, pp. 4867-4879,
Sep. 2006.

\bibitem{f}
J. Feldman, \emph{ Decoding Error-Correcting Codes via Linear Programming}, Ph.D. Thesis,
Massachusetts Institute of Technology, Sep. 2003.

\bibitem{fwk}
J. Feldman, M. J. Wainwright, and D. R. Karger, ``Using linear
programming to decode binary linear codes," \emph{IEEE Trans.
Inform. Theory}, vol. 51, no. 3, pp. 954-972, 2005.

\bibitem{hs}
J. Han and P. H. Siegel, ``Improved upper bounds on stopping redundancy,"
\emph{IEEE Trans. Inform. Theory}, vol. 53, no. 1, pp. 901-104, Jan. 2007.

\bibitem{hsv}
J. Han, P. H. Siegel, and A. Vardy, ``Improved probabilistic bounds
on stopping redundancy," \emph{IEEE Trans. Inform. Theory}, vol. 54, no. 4, pp.
1749-1753, Apr. 2008.

\bibitem{hsr}
J. Han, P. H. Siegel, and R. M. Roth, ``Single-exclusion number and the stopping redundancy of MDS codes,"
\emph{IEEE Trans. Inform. Theory}, vol. 55, no. 9, pp. 4155-4166, Sep. 2009.

\bibitem{hmlh}
T. Hehn, O. Milenkovic, S. Laendner, and J. B. Huber,
``Permutation decoding and the stopping redundancy
hierarchy of cyclic and extended cyclic codes,"
\emph{IEEE Trans. Inform. Theory}, vol. 54, no. 12, pp. 5308-5331, Dec. 2008.

\bibitem{ht2}
H. Hollmann and L. Tolhuizen, ``Erasure correcting
sets: bounds and constructions," \emph{Journal of Combinatorial
Theory, Series A}, vol. 113, pp. 1746-1759, 2006.

\bibitem{ht1}
H. Hollmann and L. Tolhuizen, ``On parity-check collections for iterative
erasure decoding that correct all correctable erasure patterns of a
given size," \emph{IEEE Trans. Inform. Theory}, vol. 53, no. 2, pp. 823-828, Feb.
2007.

\bibitem{kv}
N. Kashyap and A. Vardy, ``Stopping sets in codes from designs," in
\emph{Proc. IEEE Int. Symp. Inform. Theory}, Yokohama, Japan, Jun./Jul. 2003, p. 122.

\bibitem{koetter}
R. Koetter and P. O. Vontobel, ``Graph covers and iterative
decoding of finite-length codes," \emph{Proc. $3$rd Int. Conf.
Turbo Codes and Related Topics}, Brest, France, Sep. 2003, pp. 75-82.

\bibitem{klf}
Y. Kou, S. Lin, and M. P. C. Fossorier, ``Low-density parity-check
codes based on finite geometries: A rediscovery and new results,"
\emph{ IEEE Trans. Inform. Theory}, vol. 47, no. 7, pp. 2711-2736,
2001.

\bibitem{ks}
K. M. Krishnan and P. Shankar, ``Computing the stopping distance of
a Tanner graph is NP-hard," \emph{IEEE Trans. Inform. Theory}, vol. 53, no. 6, pp.
2278-2280, Jun. 2007.

\bibitem{lm}
S. Laendner and O. Milenkovic, ``LDPC codes based on Latin squares: cycle structure, stopping set,
and trapping set analysis,"
\emph{IEEE Trans. Communications}, Vol. 55, No. 2, pp. 303-312, Feb. 2007.

\bibitem{ms}
F. J. MacWilliams and N. J. A. Sloane, \emph{ The Theory of
Error-Correcting Codes}. Amsterdam, The Netherlands:
North-Holland, 1981 (3rd printing).

\bibitem{m}
R. J. McEliece, ``Are there turbo-codes on Mars?" Shannon
Lecture, \emph{Proc. IEEE Int. Symp. Inform. Theory}, Chicago, IL,
USA, Jun./Jul. 2004. The slides are available at the web
site http://www.systems.caltech.edu/EE/Faculty/rjm/.

\bibitem{msw}
O. Milenkovic, E. Soljanin, and P. Whiting,
``Asymptotic spectra of trapping sets in regular and irregular LDPC code ensembles,"
\emph{IEEE Trans. Inform. Theory}, vol. 53, no. 1, pp. 39-55, 2007.

\bibitem{ovz}
A. Orlitsky, K. Viswanathan, and J. Zhang, ``Stopping set
distribution of LDPC code ensembles," \emph{IEEE Trans. Inform.
Theory}, vol. 51, no. 3, pp. 929-953, Mar. 2005.

\bibitem{r}
V. Rathi, ``On the asymptotic weight and stopping set distribution
of regular LDPC ensembles," \emph{IEEE Trans. Inform. Theory},
vol. 52, no. 9, pp. 4212-4218, Sep. 2006.

\bibitem{sv}
M. Schwartz and A. Vardy, ``On the stopping distance and the
stopping redundancy of codes," \emph{IEEE Trans. Inform. Theory},
vol. 52, no. 3, pp. 922-932, 2006.

\bibitem{txla05}
H. Tang, J. Xu, S. Lin, and K. A. S.
Abdel-Ghaffar, ``Codes on finite geometries," \emph{ IEEE Trans.
Inform. Theory}, vol. 51, no. 2, pp. 572-596, 2005.

\bibitem{tanner}
R. M. Tanner, ``A recursive approach to low complexity codes,"
\emph{IEEE Trans. Inform. Theory}, vol. 27, no. 5, pp. 533-547,
Sep. 1981.

\bibitem{w}
T. Wadayama, ``Average stopping set weight distributions of redundant random ensembles,"
\emph{IEEE Trans. Inform. Theory}, vol. 54, no. 11, pp. 4991-5004, Nov. 2008.

\bibitem{wa}
J. H. Weber and K. A. S. Abdel-Ghaffar, ``Stopping set analysis
for Hamming codes," \emph{Proc. 2005 IEEE Information Theory Workshop}, Rotorua, New Zealand, Aug./Sep. 2005, pp. 244-247.

\bibitem{wa2}
J. H. Weber and K. A. S. Abdel-Ghaffar, ``Results on parity-check matrices with
optimal stopping and/or dead-end set enumerators,"
\emph{IEEE Trans. Inform. Theory}, vol. 54, no. 3, pp. 1368-1374, 2008.

\bibitem{xf}
S.-T. Xia and F.-W. Fu, ``On the minimum pseudo-codewords of LDPC
codes," \emph{IEEE Communications Letters}, vol. 10, no. 5, pp.
363-365, May 2006.

\bibitem{xf2}
S.-T. Xia and F.-W. Fu, ``On the stopping distance of finite
geometry LDPC Codes," \emph{IEEE Communications Letters}, vol.
10, no.5, pp. 381-383, May 2006.

\bibitem{xf3}
S.-T. Xia and F.-W. Fu, ``Stopping set distributions of some linear codes,"
\emph{Proc. IEEE Inform. Theory Workshop}, Chengdu, China, Oct. 2006, pp. 47-51.

\bibitem{xf1}
S.-T. Xia and F.-W. Fu, ``Minimum pseudoweight and minimum pseudocodewords of LDPC codes,"
\emph{IEEE Trans. Inform. Theory}, vol. 54, no. 1, pp. 480-485, Jan. 2008.

\end{thebibliography}
\end{document}